%% file: esfs.tex
\documentclass[pre,twocolumn,showpacs,superscriptaddress,preprintnumbers,floatfix]{revtex4}

\usepackage{ifpdf}
\usepackage{hyperref}
\usepackage{dcolumn}
\usepackage{url}
\usepackage{amsmath}
\usepackage{amscd}
\usepackage{amsfonts}
\usepackage{amssymb}
\usepackage{bm}   
\usepackage{bbm}   
\usepackage{verbatim}
\usepackage{stmaryrd}
\usepackage{amsthm}

\ifpdf
  \usepackage[pdftex]{graphicx}         
\else
  \usepackage[dvips]{graphicx}
\fi

\input{cmechabbrev}

\input{shortcuts}

\begin{document}

\title{Exact Synchronization for Finite-State Sources}

\author{Nicholas F. Travers}
\email{ntravers@math.ucdavis.edu}
\affiliation{Complexity Sciences Center}
\affiliation{Mathematics Department}

\author{James P. Crutchfield}
\email{chaos@ucdavis.edu}
\affiliation{Complexity Sciences Center}
\affiliation{Mathematics Department}
\affiliation{Physics Department\\
University of California at Davis,\\
One Shields Avenue, Davis, CA 95616}
\affiliation{Santa Fe Institute\\
1399 Hyde Park Road, Santa Fe, NM 87501}

\date{\today}

\bibliographystyle{unsrt}

\begin{abstract}
We analyze how an observer synchronizes to the internal state of a finite-state
information source, using the \eM\ causal representation. Here, we treat the
case of exact synchronization, when it is possible for the observer to synchronize
completely after a finite number of observations. The more difficult case of strictly 
asymptotic synchronization is treated in a sequel. In both cases, we find that
an observer, on average, will synchronize to the source state exponentially fast and that, as a result, 
the average accuracy in an observer's predictions of the source output approaches its optimal 
level exponentially fast as well. Additionally, we show here how to analytically calculate the
synchronization rate for exact \eMs\ and provide an efficient polynomial-time
algorithm to test \eMs\ for exactness.
\end{abstract}

\preprint{Santa Fe Institute Working Paper 10-11-031}
\preprint{arxiv.org:1008.4182 [nlin.CD]}

\maketitle


\section{Introduction}
\vspace{-0.1in}

Synchronization and state estimation for finite-state sources is a central
interest in several disciplines, including information theory, theoretical
computer science, and dynamical systems \cite{Forn05,Vite67,Jono96,Sand05a,Paz71a}.
Here, we study the synchronization problem for a class of finite-state hidden
Markov models known as \eMs\ \cite{CMechMerged}. These machines have the
important property of \emph{unifilarity}, meaning that the next state is
completely determined by the current state and the next output symbol generated.
Thus, if an observer is ever able to synchronize to the machine's internal
state, it remains synchronized forever using continued observations of the
output. As we will see, our synchronization results also have important
consequences for prediction. The future output of an \eM\ is a function of the
current state, so better knowledge of its state enables an observer to make
better predictions of the output. 


\vspace{-0.1in}
\section{Background}
\label{sec:Background}
\vspace{-0.1in}

This section provides the necessary background for our results, including
information-theoretic measures of prediction for stationary information
sources and formal definitions of \eMs\ and synchronization. In particular, we
identify two qualitatively distinct types of synchronization: exact
(synchronization via finite observation sequences) and asymptotic
(requiring infinite sequences). The exact case is the subject here; 
the nonexact case is treated in a sequel \cite{Trav10b}.

\vspace{-0.1in}
\subsection{Stationary Information Sources}
\vspace{-0.1in}

Let $\MeasAlphabet$ be a finite alphabet, and let
$\MeasSymbol_0, \MeasSymbol_1, \ldots $ be the random
variables (RVs) for a sequence of observed symbols $x_t \in \MeasAlphabet$
generated by an information source. We denote the RVs for the sequence of
future symbols beginning at time $t = 0$ as
$\Future = \MeasSymbol_0 \MeasSymbol_1 \MeasSymbol_2 ...$, the \emph{block} of
$L$ symbols beginning at time $t = 0$  as
$\Future^L = \MeasSymbol_0 \MeasSymbol_1 ... \MeasSymbol_{L-1}$, and the block
of $L$ symbols beginning at a given time $t$ as
$\Future_t^L = \MeasSymbol_t \MeasSymbol_{t+1} ... \MeasSymbol_{t+L-1}$.
A \emph{stationary source} is one for which
$\Prob(\Future_t^L) = \Prob(\Future_0^L)$ for all $t$ and all $L > 0$. 

We monitor an observer's predictions of a stationary source using
information-theoretic measures \cite{InfoThMerged}, as reviewed below. 

\begin{Def} 
The \emph{block entropy} $H(L)$ for a stationary source is:
\begin{align*}
H(L) \equiv H[\Future^L] = - \sum_{\{\future^L\}} \Prob(\future^L) \log_2 \Prob(\future^L) ~.
\end{align*}
\end{Def}
\noindent
The block entropy gives the average uncertainty in observing blocks $\Future^L$.

\begin{Def} 
The \emph{entropy rate} $\hmu$ is the asymptotic average entropy per symbol:
\begin{align*}
\hmu & \equiv \lim_{L \to \infty} \frac{H(L)}{L} \\
  & = \lim_{L \to \infty} H[\MeasSymbol_L|\Future^L] ~.
\end{align*}
\end{Def}
\noindent

\begin{Def} 
The entropy rate's \emph{length-$L$ approximation} is:
\begin{align*}
\hmu(L) & \equiv H(L) - H(L-1) \\
  & = H[\MeasSymbol_{L-1}|\Future^{L-1}] ~.
\end{align*}
\end{Def}
\noindent
That is, $\hmu(L)$ is the observer's average uncertainty in the next
symbol to be generated after observing the first $L-1$ symbols.

For any stationary process, $\hmu(L)$ monotonically decreases to the limit
$\hmu$ \cite{InfoThMerged}. However, the form of convergence depends on the 
process. The lower the value of $\hmu$ a source has, the better an observer's
predictions of the source output will be asymptotically. The faster $\hmu(L)$
converges to $\hmu$, the faster the observer's predictions reach this optimal
asymptotic level. If we are interested in making predictions after a 
finite observation sequence, then the source's true entropy rate $\hmu$, as well as
the rate of convergence of $\hmu(L)$ to $\hmu$, are both important properties of 
an information source.

\vspace{-0.2in}
\subsection{Hidden Markov Models}
\vspace{-0.1in}

In what follows we restrict our attention to an important class of stationary
information sources known as hidden Markov models. For simplicity, we assume
the number of states is finite. 

\begin{Def}
\label{Def:HMM}
A finite-state edge-label \emph{hidden Markov machine} (HMM) consists of
\begin{enumerate}
\setlength{\topsep}{0mm}
\setlength{\itemsep}{0mm}
\item a finite set of states
	$\CausalStateSet = \{\causalstate_1, ... , \causalstate_N \}$,
\item a finite alphabet of symbols $\MeasAlphabet$, and
\item a set of $N$ by $N$ symbol-labeled transition matrices
	$T^{(\meassymbol)}$, $\meassymbol \in \MeasAlphabet$,
	where $T^{(\meassymbol)}_{ij}$ is the probability of transitioning from
	state $\causalstate_i$ to state $\causalstate_j$ on symbol $\meassymbol$.
	The corresponding overall state-to-state transition matrix is denoted 
	$T = \sum_{x \in \MeasAlphabet} T^{(x)}$.
\end{enumerate}
A hidden Markov machine can be depicted as a directed graph with labeled edges.
The nodes are the states $\{\causalstate_1, ... , \causalstate_N \}$ and for
all $\meassymbol,i,j$ with $T^{(\meassymbol)}_{ij} > 0$ there is an edge from
state $\causalstate_i$ to state $\causalstate_j$ labeled $p|\meassymbol$ for
the symbol $x$ and transition probability $p = T^{(\meassymbol)}_{ij}$. We require
that the transition matrices $T^{(\meassymbol)}$ be such that this graph is
strongly connected.
\end{Def}

A hidden Markov machine $M$ generates a stationary process $\Process = (\MeasSymbol_L)_{L\geq 0}$ 
as follows. Initially, $M$ starts in some state $\causalstate_{i^*}$ chosen
according to the stationary distribution $\pi$ over machine states---the
distribution satisfying $\pi T = \pi$.
It then picks an outgoing edge according to their relative transition probabilities
$T^{(\meassymbol)}_{i^*j}$, generates the symbol $\meassymbol^*$ labeling this
edge, and follows the edge to a new state $\causalstate_{j^*}$. The next
output symbol and state are consequently chosen in a similar fashion, and this
procedure is repeated indefinitely. 

We denote $\CausalState_0, \CausalState_1, \CausalState_2, \ldots $ as the RVs
for the sequence of machine states visited and
$\MeasSymbol_0, \MeasSymbol_1, \MeasSymbol_2, \ldots$
as the RVs for the associated sequence of output symbols generated. 
The sequence of states $(\CausalState_L)_{L \geq 0}$
is a Markov chain with transition kernel $T$. However, the stochastic process
we consider is not the sequence of states, but rather the associated sequence
of outputs $(\MeasSymbol_L)_{L \geq 0}$, which generally is not Markovian.
We assume the observer directly observes this sequence of outputs, but does not
have direct access to the machine's ``hidden'' internal states.

\vspace{-0.2in}
\subsection{Examples}
\vspace{-0.1in}

In what follows, it will be helpful to refer to several example hidden Markov machines that
illustrate key properties and definitions. We introduce four examples,
all with a binary alphabet $\MeasAlphabet = \{0, 1 \}$.

\vspace{-0.1in}
\subsubsection{Even Process}
\vspace{-0.1in}

Figure~\ref{fig:EvenProcess} gives a HMM for the Even Process.
Its transitions matrices are:
\begin{align}
T^{(0)} & =
  \left(
  \begin{array}{cc}
  	p & 0 \\
	0 & 0 \\
  \end{array}
  \right) ~,\nonumber \\
T^{(1)} & =
  \left(
  \begin{array}{cc}
	0 & 1-p \\
	1 & 0 \\
  \end{array}
  \right) ~.
\end{align}

\begin{figure}[ht]
\centering
\includegraphics[scale=0.75]{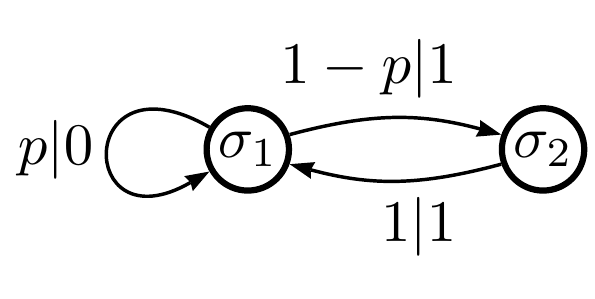}
\caption{A hidden Markov machine (the \eM) for the Even Process. The
  transitions denote the probability $p$ of generating symbol $x$ as $p|x$.
  }
\label{fig:EvenProcess}
\end{figure}

The support for the Even Process consists of all binary sequences in which
blocks of uninterrupted $1$s are even in length, bounded by $0$s. After each
even length is reached, there is a probability $p$ of breaking the block of
$1$s by inserting a $0$. The hidden Markov machine has two internal states,
$\CausalStateSet = \{ \causalstate_1, \causalstate_2 \}$, and a single
parameter $p \in (0,1)$ that controls the transition probabilities.

\vspace{-0.1in}
\subsubsection{Alternating Biased Coins Process}
\vspace{-0.1in}

Figure~\ref{fig:ABC} shows a HMM for the Alternating Biased
Coins (ABC) Process. The transitions matrices are:
\begin{align}
T^{(0)} & =
  \left(
  \begin{array}{cc}
  	0 & 1-p \\
	1-q & 0 \\
  \end{array}
  \right) ~, \nonumber \\
T^{(1)} & =
  \left(
  \begin{array}{cc}
  	0 & p \\
	q & 0 \\
  \end{array}
  \right) ~.
\end{align}

\begin{figure}[h]
\centering
\includegraphics[scale=0.75]{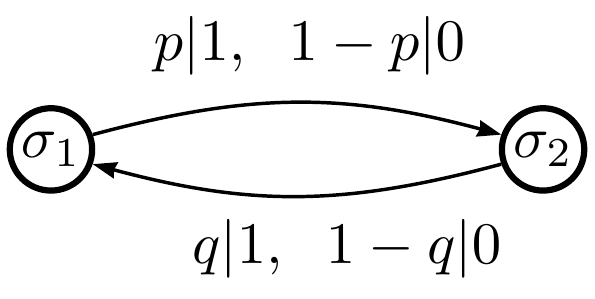}
\caption{A hidden Markov machine (the \eM) for the Alternating Biased
  Coins Process. 
  }
\label{fig:ABC}
\end{figure}

The process generated by this machine can be thought of as 
alternatively flipping two coins of different biases $p \not= q$. 

\subsubsection{SNS Process}

Figure~\ref{fig:SNS} depicts a two-state HMM for the SNS Process
which generates long sequences of $1$s broken by isolated $0$s.
Its matrices are:
\begin{align}
T^{(0)} & =
  \left(
  \begin{array}{cc}
  	0 & 0 \\
	1-q & 0 \\
  \end{array} 
  \right) ~,\nonumber\\
T^{(1)} & =
  \left(
  \begin{array}{cc}
  	p & 1-p \\
	0 & q \\
  \end{array}
  \right) ~.
\end{align}

\begin{figure}[h]
\centering
\includegraphics[scale=0.75]{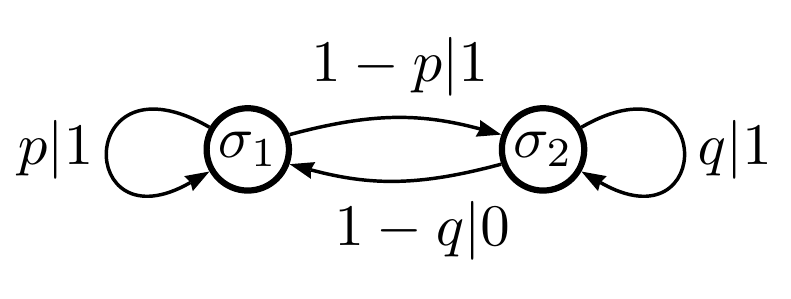}
\caption{An HMM for the SNS Process.}
\label{fig:SNS}
\end{figure} 

Note that the two transitions leaving state $\causalstate_1$ both
emit $x = 1$. 

\subsubsection{Noisy Period-$2$ Process}

Finally, Fig.~\ref{fig:NMS} depicts a nonminimal HMM for the Noisy
Period-$2$ (NP2) Process. The transition matrices are:
\begin{align}
T^{(0)} & =
  \left(
  \begin{array}{cccc}
  	0 & 0 & 0 & 0 \\
	0 & 0 & 1-p & 0 \\
	0 & 0 & 0 & 0 \\
	1-p & 0 & 0 & 0 \\
  \end{array}
  \right) ~,\nonumber \\
T^{(1)} & =
  \left(
  \begin{array}{cccc}
  	0 & 1 & 0 & 0 \\
	0 & 0 & p & 0 \\
	0 & 0 & 0 & 1 \\
	p & 0 & 0 & 0 \\
  \end{array}
  \right) ~.
\end{align}

\begin{figure}[h]
\centering
\includegraphics[scale=0.75]{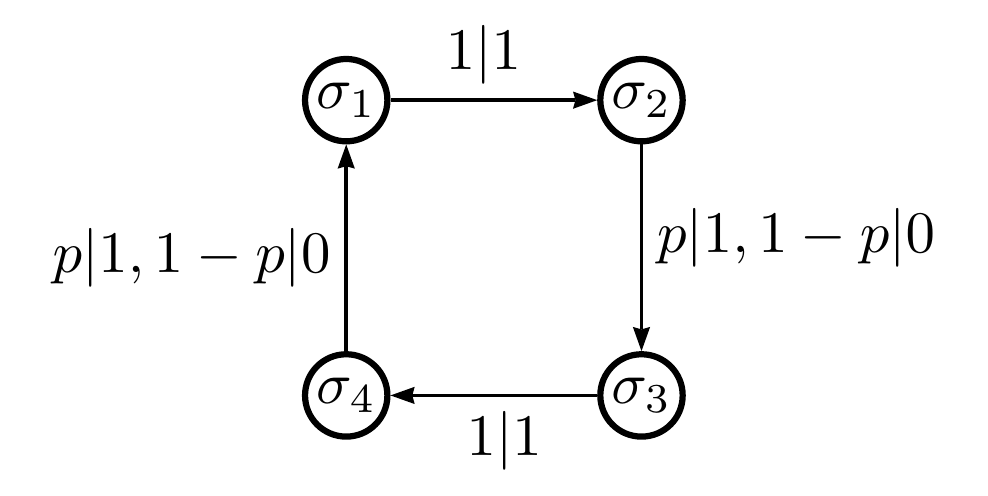}
\caption{An HMM for the Noisy Period-$2$ Process.}
\label{fig:NMS}
\end{figure}
\noindent

It is clear by inspection that the same process can be captured by a hidden
Markov machine with fewer states. Specifically, the distribution over future
sequences from states $\causalstate_1$ and $\causalstate_3$ are the same,
so those two states are redundant and can be merged. The same is also true for
states $\causalstate_2$ and $\causalstate_4$.

\subsection{\EMs}

We now introduce a class of hidden Markov machines that has a number of
desirable properties for analyzing synchronization.

\begin{Def}
A \emph{finite-state \eM} is a finite-state edge-label hidden Markov machine
with the following properties:
\begin{enumerate}
\setlength{\topsep}{0mm}
\setlength{\itemsep}{0mm}
\item \emph{Unifilarity}: For each state $\causalstate_k \in \CausalStateSet$
	and each symbol $\meassymbol \in \MeasAlphabet$ there is at most one
	outgoing edge from state $\causalstate_k$ labeled with symbol $\meassymbol$.
\item \emph{Probabilistically distinct states}: For each pair of distinct states
	$\causalstate_k, \causalstate_j \in \CausalStateSet$ there exists some
	finite word $w = \meassymbol_0 \meassymbol_1 \ldots \meassymbol_{L-1}$ such that: 
	\begin{align*}
	\Prob(\Future^L = w|\CausalState_0 = \causalstate_k)
		\not= \Prob(\Future^L = w|\CausalState_0 = \causalstate_j).
	\end{align*}
\end{enumerate} 
\label{def:eM}
\end{Def}

The hidden Markov machines given above for the Even Process and ABC Process are
both $\epsilon$-machines. The SNS machine of example 3 is not an \eM, however,
since state $\causalstate_1$ is not unifilar. The NP2 machine of example 4 is
also not an \eM, since it does not have probabilistically distinct states,
as noted before. 

\EMs\ were originally defined in Ref. \cite{CMechMerged} as hidden Markov
machines whose states, known as \emph{causal states}, were the equivalence
classes of infinite pasts $\past$ with the same probability distribution over
futures $\future$. This ``history'' \eM\ definition is, in fact, equivalent to
the ``generating'' \eM\ definition presented above in the finite-state case.
Although, this is not immediately apparent. Formally, it follows from the
synchronization results established here and in Ref. \cite{Trav10b}.

It can also be shown that an \eM\ $M$ for a given process $\Process$ is unique
up to isomorphism \cite{CMechMerged}. That is, there cannot be two different
finite-state edge-label hidden Markov machines with unifilar transitions and
probabilistically distinct states that both generate the same process
$\Process$. Furthermore, \eMs\ are minimal unifilar generators in the sense
that any other unifilar machine $M'$ generating the same process $\Process$
as an $\epsilon$-machine $M$ will have more states than $M$. Note that
uniqueness does not hold if we remove either condition 1 or 2 in Def.
\ref{def:eM}.

\subsection{Synchronization} 

Assume now that an observer has a correct model $M$ ($\epsilon$-machine) for a
process $\Process$, but is not able to directly observe $M$'s hidden internal
state. Rather, the observer must infer the internal state by observing the
output data that $M$ generates. 

For a word $w$ of length $L$ generated by $M$ let
$\phi(w) = \Prob(\CausalStateSet|w)$ be the observer's \emph{belief distribution} 
as to the current state of the machine after observing $w$. That is, 
\begin{align*}
\phi(w)_k	& = \Prob(\CausalState_L = \causalstate_k | \Future^L=w) \\
		& \equiv \Prob(\CausalState_L = \causalstate_k | \Future^L=w, \CausalState_0 \sim \pi) ~.
\end{align*}
And, define:
\begin{align*}
u(w) & = H[\phi(w)] \\
         &= H[\CausalState_L|\Future^L = w],
\end{align*} 
as the observer's uncertainty in the machine state after observing $w$. 

Denote $\LM$ as the set of all finite words that $M$ can generate, $\LLM$
as the set of all length-$L$ words it can generate, and $\LiM$ as the set of 
all infinite sequences $\future = \meassymbol_0 \meassymbol_1 ...$ that it can generate.

\begin{Def}
A word $w \in \LM$ is a \emph{synchronizing word} (or \emph{sync word)} for $M$
if $u(w) = 0$; that is, if the observer knows the current state of the machine 
with certainty after observing $w$. 
\end{Def} 

We denote the set of $M$'s infinite synchronizing sequences as $\SYN(M)$ and the set of $M$'s 
infinite weakly synchronizing sequences as $\WSYN(M)$:
\begin{align*}
& \SYN(M) = \{ \future \in \mathcal{L}_{\infty}(M) : u(\future^L) = 0 \mbox{ for some } L\} ,~ and \\
& \WSYN(M) = \{ \future \in \LiM : u(\future^L) \rightarrow 0 \mathrm{~as~} L \rightarrow \infty \} ~.
\end{align*} 

\begin{Def}
\label{Def:ExactSync}
An \eM\ M is \emph{exactly synchronizable} (or simply \emph{exact}) if
$\Prob(\SYN(M)) = 1$; that is, if the observer synchronizes to almost every
(a.e.) sequence generated by the machine in finite time.
\end{Def}

\begin{Def}
\label{Def:AsymptoticSync}
An \eM\ M is \emph{asymptotically synchronizable} if $\Prob(\WSYN(M)) = 1$;
that is, if the observer's uncertainty in the machine state vanishes asymptotically
for a.e. sequence generated by the machine. 
\end{Def}

The Even Process \eM, Fig. \ref{fig:EvenProcess}, is an exact machine.
Any word containing a $0$ is a sync word for this machine, 
and almost every $\future$ it generates contains at least one $0$.
The ABC Process \eM, Fig. \ref{fig:ABC}, is not exactly synchronizable,
but it is asymptotically synchronizable.

\begin{Rem}
If $w \in \LM$ is a sync word, then by unifilarity so is $wv$, for all
$v$ with $wv \in \LM$. Once an observer synchronizes exactly, it remains
synchronized exactly for all future times. It follows that any exactly
synchronizable machine is also asymptotically synchronizable. 
\end{Rem}

\begin{Rem}
If $w \in \LM$ is a sync word then so is $vw$, for all $v$ with 
$vw \in \LM$. Since any finite word $w \in \LM$ will be contained in almost
every infinite sequence $\future$ the machine generates, it follows that
a machine is exactly synchronizable  if (and only if) it has some sync word
$w$ of finite length.
\end{Rem}

\begin{Rem}
It turns out all finite-state \eMs\ are asymptotically synchronizable;
see Ref. \cite{Trav10b}. Hence, there are two disjoint classes to consider:
\emph{exactly synchronizable machines} and \emph{asymptotically synchronizable
machines that are nonexact}. The exact case is the subject of the remainder.
\end{Rem}

Finally, one last important quantity for synchronization is the observer's
average uncertainty in the machine state after seeing a length-$L$ block of
output \cite{InfoThMerged}.

\begin{Def}
The observer's \emph{average state uncertainty} at time $L$ is:
\begin{align}
\AvgUncertainty(L) & \equiv H[\CausalState_L|\Future^L] \nonumber \\
  & = \sum_{\{ \future^L \}} \Prob(\future^L)
  \cdot H[\CausalState_L|\Future^L = \future^L] ~.
\end{align}
\end{Def}
That is, $\AvgUncertainty(L)$ is the expected value of an observer's uncertainty
in the machine state after observing $L$ symbols.

Now, for an \eM, an observer's prediction of the next output symbol is a direct
function of the probability distribution over machine states induced by the
previously observed symbols. Specifically, 
\begin{align}
\Prob(X_L & = \meassymbol | \Future^L = \future^L) \nonumber \\
  & = \sum_{ \{\causalstate_k\} } \Prob(\meassymbol |\causalstate_k)
  \cdot \Prob(\CausalState_L = \causalstate_k | \Future^L = \future^L) ~.
\end{align}
Hence, the better an observer knows the machine state at the current time,
the better it can predict the next symbol generated. And, on average, the
closer $\AvgUncertainty(L)$ is to 0, the closer $\hmu(L)$ is to $\hmu$.
Therefore, the rate of convergence of $\hmu(L)$ to $\hmu$ for an \eM\ is
closely related to the average rate of synchronization.


\section{Exact Synchronization Results}
\label{sec:Results}

This section provides our main results on synchronization rates for exact
machines and draws out consequences for the convergence rates of
$\AvgUncertainty(L)$ and $\hmu(L)$. 

The following notation will be used throughout:
\begin{itemize}
\setlength{\topsep}{0mm}
\setlength{\itemsep}{0mm}
\item $\SYN_L = \{ w \in \LLM: w$ is a sync word for $M \}$. 
\item $\NSYN_L = \{ w \in \LLM: w$ is \emph{not} a sync word for $M \}$. 
\item $\SYN_{L,\causalstate_k} = \{ w \in \LLM: w$ synchronizes the
	observer to state $\causalstate_k \}$. 
\item $\L(M,\causalstate_k) = \{ w: M$ can generate $w$ starting in
	state $\causalstate_k \}$. 
\item For words $w,w' \in \L(M)$, we say $w \subset w'$ if there exist
          words $u,v$ (of length $\geq 0$) such that $w' = uwv$. 
\item For a word $w \in \L(M,\causalstate_k)$,  $\delta(\causalstate_k,w)$
	is defined to be the (unique) state in $\CausalStateSet$ such that
	$\causalstate_k \goesonw  \delta(\causalstate_k,w)$. 
\item For a set of states $S \subset \CausalStateSet$, we define: 
\begin{align*}
\delta(S,w) = \{ \causalstate_j \in \CausalStateSet:
  \causalstate_k \goesonw \causalstate_j
  \mbox{ for some } \causalstate_k \in S 
  \} ~.
\end{align*}
\end{itemize}

\subsection{Exact Machine Synchronization Theorem}

Our first theorem states that an observer synchronizes (exactly) to the
internal state of any exact \eM\  exponentially fast.  
\begin{The}
For any exact \eM\ $M$, there are constants $K > 0$ and $0 < \alpha < 1$
such that:
\begin{equation}
\Prob(\NSYN_L) \leq K \alpha^L ~,
\label{Eq:ExactMachineSync}
\end{equation}
for all $L \in \N$.
\label{Thm:ExactMachineSync}
\end{The}
\begin{proof}
Let $M$ be an exact machine with sync word $w \in \L(M, \causalstate_k)$.
Since the graph of $M$ is strongly connected, we know that for each state
$\causalstate_j$ there is a word $v_j$ such that
$\delta(\causalstate_j,v_j) = \causalstate_k$.
Let $w_j = v_jw$, $n = \max_j |w_j|$, and $p = \min_j \Prob(w_j|\causalstate_j)$.
Then, for all $L \geq 0$, we have: 
\begin{align}
\Prob(w \subset \Future^{n + L} & | w \not\subset \Future^L) \nonumber \\
  & \geq \Prob(w \subset \Future^n_L | w \not\subset \Future^L)
  \nonumber \\
  & \geq \min_j \Prob(w \subset \Future^n_L | \CausalState_L = \causalstate_j)
  \nonumber \\
  & \geq p ~.
\end{align}
Hence,
\begin{align}
\Prob(w \not\subset \Future^{n + L} | w \not\subset \Future^L)
  \leq 1-p ~,
\end{align}
for all $L \geq 0$.
And, therefore, for all $m \in \N$:
\begin{align}
\Prob(\NSYN_{mn}) & \leq Pr(w \not\subset \Future^{mn}) \nonumber \\
  & = \Prob(w \not\subset \Future^{n})
  \cdot \Prob(W \not\subset \Future^{2n}|w \not\subset \Future^n)
  \nonumber \\
  & ~~~~~\cdots \Prob(w \not\subset \Future^{mn}|w \not\subset \Future^{(m-1)n})
  \nonumber \\
  & \leq (1-p) \cdot (1-p) \cdot \cdot \cdot (1-p) \nonumber \\
  & = (1-p)^m \nonumber \\
  & = \beta^m ~,
\end{align}
where $\beta \equiv 1-p$. 
Or equivalently, for any length $L = mn$ ($m \in \N$):
\begin{align}
\Prob(NSYN_L) \leq \alpha^L ~,
\end{align}
where $\alpha \equiv \beta^{1/n}$. Since $\Prob(\NSYN_L)$ is monotonically
decreasing, it follows that:
\begin{align}
\Prob(\NSYN_L) \leq \frac{1}{\alpha^n} \cdot \alpha^L = K \alpha^L ~,
\end{align}
for all $L \in \N$, where $K \equiv 1/\alpha^n$.  
\end{proof}
\begin{Rem}
In the above proof we implicitly assumed $\beta \not=0$.
If $\beta = 0$, then the conclusion follows trivially.
\end{Rem}

\subsection{Synchronization Rate}
\label{SubSec:SyncRate}

Theorem \ref{Thm:ExactMachineSync} states that an observer synchronizes
(exactly) to any exact \eM\ exponentially fast. However, the sync rate constant:
\begin{equation}
\alpha^* = \lim_{L \to \infty} \Prob(\NSYN_L)^{1/L} \\
\end{equation}
depends on the machine, and it may often be of practical interest to know the
value of this constant. We now provide a method for computing $\alpha^*$
analytically. It is based on the construction of an auxiliary machine $\Mt$. 

\begin{Def}
Let $M$ be an \eM\ with states
$\CausalStateSet = \{\causalstate_1, ... , \causalstate_N \}$, 
alphabet $\MeasAlphabet$, and transition matrices
$T^{(x)}, x \in \MeasAlphabet$. 
The \emph{possibility machine} $\Mt$ is defined as follows:
\begin{enumerate}
\item The alphabet of $\Mt$ is $\MeasAlphabet$.
\item The states of $\Mt$ are pairs of the form $(\causalstate,S)$ where
	$\causalstate \in \CausalStateSet$ and $S$ is a subset of $\CausalStateSet$
	that contains $\causalstate$. 
\item The transition probabilities are:
\begin{align*}
\Prob( (\causalstate,S) & \goesonx (\causalstate',S')) \\
  & = \Prob(x|\causalstate) I(x,(\causalstate,S),(\causalstate',S')) ~,
\end{align*}
where $I(x,(\causalstate,S),(\causalstate',S'))$ is the indicator function:
\begin{align*}
I(x, & (\causalstate,S),(\causalstate',S')) \\
  & = \left\{ \begin{array}{ll}
  	1 & \mbox{if } \delta(\causalstate,x) = \causalstate'
  		\mbox{ and } \delta(S,x) = S' \\
  	0 & \mbox{otherwise.}
  \end{array} \right.
\end{align*}
\end{enumerate}
A state of $\Mt$ is said to be \emph{initial} if it is of the form
$(\causalstate,\CausalStateSet)$ for some $\causalstate \in \CausalStateSet$.
For simplicity we restrict the $\Mt$ machine to consist of only those states
that are accessible from initial states. The other states are irrelevant for
the analysis below. 
\end{Def}

The idea is that $\Mt$'s states represent states of the joint (\eM, observer)
system. State $\causalstate$ is the true \eM\ state at the current
time, and $S$ is the set of states that the observer believes are currently
possible for the \eM\ to be in, after observing all previous symbols. 
Initially, all states are possible (to the observer), so initial states are
those in which the set of possible states is the complete set $\CausalStateSet$.

If the current true \eM\ state is $\causalstate$, and then the symbol $x$ is
generated, the new true \eM\ state must be $\delta(\causalstate,x)$.
Similarly, if the observer believes any of the states in $S$ are currently
possible, and then the symbol $x$ is generated, the new set of possible states
to the observer is $\delta(S,x)$. This accounts for the transitions in $\Mt$
topologically. The probability of generating a given symbol $x$ from
$(\causalstate,S)$ is, of course, governed only by the true state $\causalstate$
of the \eM\: $\Prob(x|(\causalstate,S)) = \Prob(x|\causalstate)$. 
 
An example of this construction for a $3$-state exact \eM\ is given in Appendix
\ref{AppendixA}. Note that the graph of the $\Mt$ machine there has a single
recurrent strongly connected component, which is isomorphic to the original
machine $M$. This is not an accident. It will always be the case, as long as
the original machine $M$ is exact. 
 
\begin{Rem}
If $M$ is an exact machine with more than 1 state the graph of $\Mt$ itself is never 
strongly connected. So, $\Mt$ is not an \eM\ or even an HMM in the sense of Def.
\ref{Def:HMM}. However, we still refer to $\Mt$ as a ``machine''. 
\end{Rem}
 
In what follows, we assume $M$ is exact. We denote the states of $\Mt$ as
$\widetilde{\CausalStateSet } = \{ q_1,\ldots,q_{\Nt} \}$, its symbol-labeled
transition matrices $\Tt^{(x)}$, and its overall state-to-state transition 
matrix $\Tt = \sum_{x \in \MeasAlphabet} \Tt^{(x)}$. We assume the states are
ordered in such a way that the initial states
$(\causalstate_1,\CausalStateSet), \ldots , (\causalstate_N,\CausalStateSet)$
are, respectively, $q_1, ... , q_N$. Similarly, the \emph{recurrent states} 
$(\causalstate_1,\{\causalstate_1\}), (\causalstate_2,\{\causalstate_2\}),
\ldots , (\causalstate_N,\{\causalstate_N\})$ are, respectively,
$q_{n + 1}, q_{n + 2}, ... ,q_{\Nt}$, where $n = \Nt - N$. The ordering of the
other states is irrelevant. In this case, the matrix $\Tt$ has the following
block upper-triangular form:
\begin{align}
\Tt =  \left( \begin{array}{cc}
B & B'\\
O & T
\end{array} \right) ~,
\end{align}

\noindent
where $B$ is a $n \times n$ matrix with nonnegative entries, $B'$ is a
$n \times N$ matrix with nonnegative entries, $O$ is a $N \times n$ matrix
of all zeros, and $T$ is the $N \times N$ state-to-state transition matrix of
the original \eM\ $M$. 

Let $\pit = (\pi_1, ... , \pi_N, 0, ... , 0)$ denote the length-$\Nt$ row
vector whose distribution over the initial states is the same as the
stationary distribution $\pi$ for the \eM\ $M$. Then, the initial probability
distribution $\phit_0$ over states of the joint (\eM, observer) system
is simply:
\begin{equation}
\phit_0 = \pit ~,
\end{equation}
\noindent
and, thus, the distribution over states of the joint system after the first $L$
symbols is:
\begin{equation}
\phit_L = \pit \Tt^L ~.
\end{equation}
If the joint system is in a recurrent state of the form
$(\causalstate_k, \{\causalstate_k\})$, then to the observer the only possible
state of the \eM\ is the true state, so the observer is synchronized. For all
other states of $\Mt$, the observer is not yet synchronized.  Hence, the
probability the observer is not synchronized after $L$ symbols is simply the
combined probability of all nonreccurent states $q_i$ in the distribution
$\phit_L$. Specifically, we have:
\begin{align}
\Prob(\NSYN_L)	& = \sum_{i = 1}^{n} (\phit_L)_i \nonumber \\
			& = \sum_{i = 1}^{n} (\pit \Tt^L)_i \nonumber \\
			& = \sum_{i = 1}^{n} (\pi^B B^L)_i \nonumber \\
			& = \norm{\pi^B B^L}_1 ~,
\end{align}
\noindent
where $\pi^B = (\pi_1, ... , \pi_N, 0, ... , 0)$ is the length-$n$ row vector
corresponding to the distribution over initial states $\pi$. The third
equality follows from the block upper-triangular form of $\Tt$.

Appendix \ref{AppendixB} shows that:
\begin{equation}
\lim_{L \to \infty}  \left(\norm{\pi^B B^L}_1\right)^{1/L} = r ~,
\label{Eq:VectorNormLimit}
\end{equation}
where $r = r(B)$ is the (left) spectral radius of $B$:
\begin{equation}
r(B) = \max \{|\lambda| : \lambda \mbox{ is a (left) eigenvalue of } B \}.
\end{equation} 
Thus, we have established the following result.
\begin{The}
For any exact \eM\ $M$, $\alpha^* = r$. 
\label{Thm:SyncRate}
\end{The}

\subsection{Consequences}

We now apply Thm. \ref{Thm:ExactMachineSync} to show that an observer's average
uncertainty $\AvgUncertainty(L)$ in the machine state and average uncertainty
$\hmu(L)$ in predictions of future symbols both decay exponentially fast to
their respective limits: $0$ and $\hmu$. The decay constant $\alpha$ in both
cases is essentially bounded by the sync rate constant $\alpha^*$ from Thm.
\ref{Thm:SyncRate}.

\begin{Prop}
\label{Pr:AvgUncertaintyL_Convergence}
For any exact  \eM\ $M$, there are constants $K > 0$ and $0 < \alpha < 1$
such that:
\begin{align}
\AvgUncertainty(L) \leq K \alpha^L, \mbox{ for all } L \in \N.
\end{align}
\end{Prop}

\begin{proof}
Let $M$ be any exact machine. By Thm. \ref{Thm:ExactMachineSync} there
are constants $C > 0$ and $0 < \alpha < 1$ such that
$\Prob(\NSYN_L) \leq C \alpha^L$, for all $L \in \N$. Thus, we have:
\begin{align}
\AvgUncertainty(L) 
  & = \sumw \Prob(w) u(w) \nonumber \\
  & = \sum_{w \in \SYN_L} \Prob(w) u(w)
  + \sum_{w \in \NSYN_L} \Prob(w) u(w) \nonumber \\
  & \leq 0 + \sum_{w \in \NSYN_L} \Prob(w) \log(N) \nonumber \\ 
  & \leq \log(N) \cdot C  \alpha^L \nonumber \\
  & = K \alpha^L ~,
\end{align}
where $N$ is the number of machine states and $K \equiv C \log(N)$. \\
\end{proof}

Let $h_k \equiv H[\MeasSymbol_0|\CausalState_0 = \causalstate_k]$ and 
$h_w \equiv H[\MeasSymbol_0|\CausalState_0 \sim \phi(w)]$, be the 
conditional entropies in the next symbol given the state $\causalstate_k$ and 
word $w$. 

\begin{Prop}
\label{Pr:hmuL_Convergence}
For any exact \eM\ $M$:
\begin{align}
\hmu = H[\MeasSymbol_0|\CausalState_0] \equiv \sum_k \pi_k h_k
\label{eq:MachineEntropyRate}
\end{align}
and there are constants $K > 0$ and $0 < \alpha < 1$ such that:
\begin{align}
\hmu(L) - \hmu \leq K \alpha^L 	~, \mbox{ for all } L \in \N.
\end{align}
\end{Prop}

\begin{Rem}
The $\hmu$ formula Eq. (\ref{eq:MachineEntropyRate}) has been known for
some time, although in slightly different contexts. Shannon, for example,
derived this formula in his original publication \cite{Shan48a} for a type of
hidden Markov machine that is similar (apparently unifilar) to an \eM.
\end{Rem}

\begin{proof}
Let $M$ be any exact machine. Since we know $\hmu(L) \searrow \hmu$ it suffices
to show there are constants $K > 0$ and $0 < \alpha < 1$ such that:
\begin{align}
\left| \hmu(L) - \sum_k \pi_k h_k \right|
  \leq K \alpha^L ~,
\end{align}
for all $L \in \N$.
This will establish both the value of $\hmu$ and the necessary convergence.

Now, by Thm. \ref{Thm:ExactMachineSync}, there are constants $C > 0$ and
$0 < \alpha < 1$ such that $\Prob(\NSYN_L) \leq C \alpha^L$, for all
$L \in \N$. Also, note that for all $L$ and $k$ we have:
\begin{align}
\pi_k & = \sum_{w \in \LLM} \Prob(w) \cdot \phi(w)_k \nonumber \\
       	& \geq \sum_{w \in \SYN_{L,\causalstate_k}} \Prob(w) \cdot \phi(w)_k   \nonumber \\
  	& = \Prob(\SYN_{L,\causalstate_k}) ~.
\end{align}
Thus, 
\begin{align}
\sum_k \left( \pi_k - \Prob(\SYN_{L,\causalstate_k}) \right) \cdot h_k \geq 0
\label{eq:Bound1}
\end{align}
and
\begin{align}
\sum_k & ( \pi_k - \Prob(\SYN_{L,\causalstate_k}))
  \cdot h_k \nonumber \\
& \leq \sum_k (\pi_k - \Prob(\SYN_{L,\causalstate_k}))
  \cdot \log |\MeasAlphabet| \nonumber \\
& = \log |\MeasAlphabet| \cdot \left (\sum_k \pi_k - \sum_k
  \Prob(\SYN_{L,\causalstate_k}) \right) \nonumber \\
& = \log |\MeasAlphabet| \cdot (1 - \Prob(\SYN_L)) \nonumber \\
& = \log |\MeasAlphabet| \cdot \Prob(\NSYN_L) \nonumber \\
& \leq \log |\MeasAlphabet| \cdot C \alpha^L. 
\end{align}
Also, clearly, 
\begin{align}
\sum_{w \in \NSYN_L} \Prob(w) \cdot h_w \geq 0
\end{align} 
and
\begin{align} 
\sum_{w \in \NSYN_L} \Prob(w) \cdot h_w
  & \leq \log|\MeasAlphabet| \cdot \Prob(\NSYN_L) \nonumber \\
  & \leq \log|\MeasAlphabet| \cdot C \alpha^L ~. 
\label{eq:Bound4}
\end{align}

Therefore, we have for all $L \in \N$:
\begin{align}
& \left| \hmu(L+1) - \sum_k \pi_k h_k \right|
  \nonumber \\
& = \left| \sum_{w \in \LLM} \Prob(w) h_w - \sum_k \pi_k h_k \right|
  \nonumber \\
& = \left| \sum_{w \in \NSYN_L} \Prob(w) h_w + \sum_{w \in SYN_L} \Prob(w) h_w - \sum_k \pi_k h_k \right|
  \nonumber \\
& = \left| \sum_{w \in \NSYN_L} \Prob(w) h_w + \sum_k  \Prob(\SYN_{L,\causalstate_k}) h_k - \sum_k \pi_k h_k \right|
  \nonumber \\
& = \left| \sum_{w \in \NSYN_L} \Prob(w) h_w - \sum_k (\pi_k - \Prob(\SYN_{L,\causalstate_k})) h_k \right| 
\nonumber \\
& \leq C \log|\MeasAlphabet| \alpha^L ~.
\end{align}

The last inequality follows from Eqs. (\ref{eq:Bound1})-(\ref{eq:Bound4}),
since $|x-y| \leq z$ for all nonnegative real numbers $x$, $y$, and $z$ with
$x \leq z$ and $y \leq z$.

Finally, since:
\begin{align}
\left| \hmu(L+1) - \sum_k \pi_k h_k \right| \leq C \log|\MeasAlphabet| \alpha^L ~,
\end{align} 
for all $L \in \N$, we know that:
\begin{align}
\left| \hmu(L) - \sum_k \pi_k h_k \right| \leq K \alpha^L ~ ,
\end{align} 
for all $L \in \N$, where
$K \equiv (\log|\MeasAlphabet| / \alpha) \cdot \max \{C,1\}$.
\end{proof}

\begin{Rem}
For any $\alpha > \alpha^*$ there exists some $K > 0$ for which Eq.
(\ref{Eq:ExactMachineSync}) holds. Hence, by the constructive proofs above,
we see that the constant $\alpha$ in Props. \ref{Pr:AvgUncertaintyL_Convergence}
and \ref{Pr:hmuL_Convergence} can be chosen arbitrarily close to $\alpha^*$:
$\alpha = \alpha^* + \epsilon$.  
\end{Rem}

\section{Characterization of Exact \EMs}
\label{sec:ExactnessTest}

In this section we provide a set of necessary and sufficient conditions for
exactness and an algorithmic test for exactness based upon these conditions.

\subsection{Exact Machine Characterization Theorem}

\begin{Def}
States $\causalstate_k$ and $\causalstate_j$ are said to be
\emph{topologically distinct} if
$\L(M,\causalstate_k) \not= \L(M,\causalstate_j)$. 
\end{Def}

\begin{Def}
States $\causalstate_k$ and $\causalstate_j$ are 
said to be \emph{path convergent} if there exists
$w \in \L(M,\causalstate_k) \cap \L(M,\causalstate_j)$
such that $\delta(\causalstate_k,w) = \delta(\causalstate_j,w)$. 
\end{Def}

If states $\causalstate_k$ and $\causalstate_j$ are topologically distinct (or
path convergent) we will also say the pair $(\causalstate_k, \causalstate_j)$
is topologically distinct (or path convergent). 

\begin{The}
An \eM\ $M$ is exact if and only if every pair of distinct states
$(\causalstate_k, \causalstate_j)$ satisfies at least one of the following two conditions:
\begin{enumerate}
\setlength{\topsep}{0mm}
\setlength{\itemsep}{0mm}
\item[(i)] {The pair ($\causalstate_k, \causalstate_j$) is topologically distinct.}
\item[(ii)] {The pair ($\causalstate_k, \causalstate_j$) is path convergent.}
\end{enumerate}
\label{thm:EMExact}
\end{The}

\begin{proof}
It was noted above that an \eM\ $M$ is exact if and only if it has some sync
word $w$ of finite length. Therefore, it is enough to show that every pair of
distinct states ($\causalstate_k, \causalstate_j$) satisfies either (i) or
(ii) if and only if $M$ has some sync word $w$ of finite length.

We establish the ``if'' first:
\emph{If $M$ has a sync word $w$, then every pair of distinct states
($\causalstate_k, \causalstate_j$) satisfies either (i) or (ii)}.

Let $w$ be a sync word for $M$. Then $w \in \L(M,\causalstate_k)$ for some $k$.
Take words $v_j$, $j = 1,2, ... N$, such that
$\delta(\causalstate_j,v_j) = \causalstate_k$.
Then, the word $v_j w \equiv w_j \in \L(M,\causalstate_j)$ is
also a sync word for $M$ for each $j$. Therefore, for each $i \not= j$ either
$w_j \not\in \L(M,\causalstate_i)$ or
$\delta(\causalstate_i,w_j) = \delta(\causalstate_j,w_j)$.
This establishes that the pair ($\causalstate_i, \causalstate_j$) is
either topologically distinct or path convergent. Since this holds
for all $j = 1,2, ... , N$ and for all $i \not= j$, we know every pair of
distinct states is either topologically distinct or path convergent.

Now, for the ``only if'' case:
\emph{If every pair of distinct states ($\causalstate_k, \causalstate_j$)
satisfies either (i) or (ii), then $M$ has a sync word $w$}.

If each pair of distinct states ($\causalstate_k, \causalstate_j$) satisfies
either (i) or (ii), then for all $k$ and $j$ ($k \not= j$) there is some word
$w_{\causalstate_k, \causalstate_j}$ such that one of the following 
three conditions is satisfied:
\begin{enumerate}
\setlength{\topsep}{0mm}
\setlength{\itemsep}{0mm}
\item $w_{\causalstate_k, \causalstate_j} \in \L(M,\causalstate_k)$, but $w_{\causalstate_k, \causalstate_j} \not\in \L(M,\causalstate_j)$.
\item $w_{\causalstate_k, \causalstate_j} \in \L(M,\causalstate_j)$,
	but $w_{\causalstate_k, \causalstate_j} \not\in \L(M,\causalstate_k)$.
\item $w_{\causalstate_k, \causalstate_j} \in \L(M,\causalstate_k) \cap
	\L(M,\causalstate_j)$
	and $\delta(\causalstate_k,w_{\causalstate_k, \causalstate_j})
	= \delta(\causalstate_j,w_{\causalstate_k, \causalstate_j})$.
\end{enumerate}

We construct a sync word $w = w_1w_2 ... w_m$ for $M$, where each
$w_i = w_{\causalstate_{k_i}, \causalstate_{j_i}}$ for some $k_i$
and $j_i$, as follows.
\begin{itemize}
\setlength{\topsep}{0mm}
\setlength{\itemsep}{0mm}
\item Let $\CausalStateSet^0
	= \{ \causalstate_1^0, \ldots, \causalstate_{N_0}^0 \}
	\equiv \CausalStateSet = \{ \causalstate_1, ... , \causalstate_N \}$.
	Take $w_1 = w_{\causalstate_1^0, \causalstate_2^0}$.
\item Let $\CausalStateSet^1 = \{ \causalstate_1^1, \ldots,
	\causalstate_{N_1}^1 \}
	\equiv \delta(\CausalStateSet^0,w_1)$.
	Since $w_1 = w_{\causalstate_1^0, \causalstate_2^0}$ satisfies
	either condition (1), (2), or (3), we know $N_1 < N_0$.
	Take $w_2 = w_{\causalstate_1^1, \causalstate_2^1}$.
\item Let $\CausalStateSet^2 = \{ \causalstate_1^2, \ldots,
	\causalstate_{N_2}^2 \}
	\equiv \delta(\CausalStateSet^1,w_2)$.
	Since $w_2 = w_{\causalstate_1^1, \causalstate_2^1}$ satisfies either
	condition (1), (2), or (3) we know $N_2 < N_1$.
	Take $w_3 = w_{\causalstate_1^2, \causalstate_2^2}$.
\end{itemize}

\begin{centering}
. \\
. \\
. \\ 
\end{centering} 
\vspace{5 mm}
Repeat until $|\CausalStateSet^m| = 1$ for some $m$.
Note that this must happen after a finite number of steps 
since $N = N_0$ is finite and $N_0 > N_1 > N_2 > \cdots$.

By this construction $w = w_1w_2 ... w_m \in \L(M)$ is a sync word for
$M$. After observing $w$, an observer knows the machine must be in state
$\causalstate_1^m$.
\end{proof}

\subsection{A Test for Exactness}

We can now provide an algorithmic test for exactness using the characterization
theorem of exact machines. We begin with subalgorithms to test for topological
distinctness and path convergence of state pairs. Both are essentially the same
algorithm and only a slight modification of the deterministic finite-automata
(DFA) table-filling algorithm to test for pairs of equivalent states
\cite{Hopc07a}.

\begin{Alg}
Test States for Topological Distinctness.
\begin{enumerate}
\setlength{\topsep}{0mm}
\setlength{\itemsep}{0mm}
\item \emph{Initialization}:
Create a table containing boxes for all pairs of distinct states
$(\causalstate_k, \causalstate_j)$. Initially, all boxes are blank.
Then,\\

Loop over distinct pairs ($\causalstate_k, \causalstate_j$) \\ 
\indent \hspace{5 mm} Loop over $\meassymbol \in \MeasAlphabet$ \\
\indent \hspace{10 mm} If $\{ \meassymbol \in \L(M,\causalstate_k)$ but
	$\meassymbol \not\in \L(M, \causalstate_j) \}$ \\
\indent \hspace{15 mm}
	or $\{ \meassymbol \in \L(M,\causalstate_j)$
	but $\meassymbol \not\in \L(M,\causalstate_k) \}$, \\
\indent \hspace{10 mm}
	then mark box for pair $(\causalstate_k, \causalstate_j)$. \\
\indent \hspace{5 mm} end \\
\indent end
\item \emph{Induction}:
If $\delta(\causalstate_k,\meassymbol) = \causalstate_{k'}$,
$\delta(\causalstate_j,\meassymbol) = \causalstate_{j'}$,
and the box for pair $(\causalstate_{k'},\causalstate_{j'})$ is already marked,
then mark the box for pair $(\causalstate_k, \causalstate_j)$. Repeat until
no more inductions are possible. \end{enumerate}
\label{Alg:TopologicalDistinctness}
\end{Alg}
 
\begin{Alg}
Test States for Path Convergence.\\

This algorithm is identical to Algorithm 1 except that the if-statement in
the initialization step is replaced with the following:
\begin{quote}
If $\meassymbol \in \L(M,\causalstate_k) \cap  \L(M,\causalstate_j)$ and
$\delta(\causalstate_k,\meassymbol) = \delta(\causalstate_j,\meassymbol)$,
then mark box for pair $(\causalstate_k, \causalstate_j)$.
\end{quote}
\label{Alg:PathConvergence}
\end{Alg}

With Algorithm \ref{Alg:TopologicalDistinctness} all pairs of topologically
distinct states end up with marked boxes. With Algorithm
\ref{Alg:PathConvergence} all pairs of path convergent states end up with
marked boxes. These facts can be proved, respectively, by using induction on
the length of the minimal distinguishing or path converging word $w$ for a
given pair of states. The proofs are virtually identical to the proof of the
standard DFA table filling algorithm, so the details have been omitted.  

Note also that both of these are polynomial-time algorithms.
Step (1) has run time $O(|\MeasAlphabet| N^2)$. The inductions in Step (2),
if done in a reasonably efficient fashion, can also be completed in run time
$O(|\MeasAlphabet| N^2)$. (See, e.g., the analysis of DFA table filling
algorithm in Ref. \cite{Hopc07a}.) Therefore, the total run time of these
algorithm is $O(|\MeasAlphabet| N^2)$. 

\begin{Alg}
Test for Exactness.
\begin{enumerate}
\setlength{\topsep}{0mm}
\setlength{\itemsep}{0mm}
\item {Use Algorithm 1 to find all pairs of topologically distinct states.}
\item {Use Algorithm 2 to find all pairs of path convergent states.}
\item {Loop over all pairs of distinct states $(\causalstate_k, \causalstate_j$)
	to check if they are either (i) topologically distinct or (ii) path convergent. 
	By Thm. \ref{thm:EMExact}, if all distinct pairs of states satisfy (i) or (ii) 
	or both, the machine is exact, and otherwise it is not.}
\end{enumerate}
\end{Alg}

This, too, is a polynomial-time algorithm.
Steps (1) and (2) have run time $O(|\MeasAlphabet| N^2)$. Step (3) has run
time $O(N^2)$. Hence, the total run time for this algorithm is
$O(|\MeasAlphabet| N^2)$. 


\section{Conclusion}
\label{sec:Conclusion}

We have analyzed the process of exact synchronization to finite-state \eMs. In
particular, we showed that for exact machines an observer synchronizes
exponentially fast. As a result, the average uncertainty $\hmu(L)$ in an
observer's predictions converges exponentially fast to the machine's entropy
rate $\hmu$---a phenomenon first reported for subshifts estimated from maps of
the interval \cite{Crut83a}. Additionally, we found an efficient (polynomial-time) algorithm to test \eMs\ for exactness. 

In Ref. \cite{Trav10b} we similarly analyze asymptotic synchronization to
nonexact \eMs. It turns out that qualitatively similar results hold. That is,
$\AvgUncertainty(L)$ and $\hmu(L)$ both converge to their respective limits
exponentially fast. However, the proof methods in the nonexact case are
substantially different. 

In the future we plan to extend these results to more generalized model classes,
such as to \eMs\ with a countable number of states and to nonunifilar hidden
Markov machines. 


\section*{Acknowledgments}
NT was partially supported on a VIGRE fellowship.
The work was partially supported by the Defense Advanced Research Projects
Agency (DARPA) Physical Intelligence project via subcontract No. 9060-000709.
The views, opinions, and findings here are those of the authors and should not
be interpreted as representing the official views or policies, either expressed
or implied, of the DARPA or the Department of Defense.


\appendix

\section{}
\label{AppendixA}

We construct the possibility machine $\Mt$ for the three-state \eM\ shown in
Fig. \ref{fig:ThreeStateEM}. The result is shown in Fig. \ref{fig:TP_Machine}. \\

\begin{figure}[ht]
\includegraphics[scale=0.65]{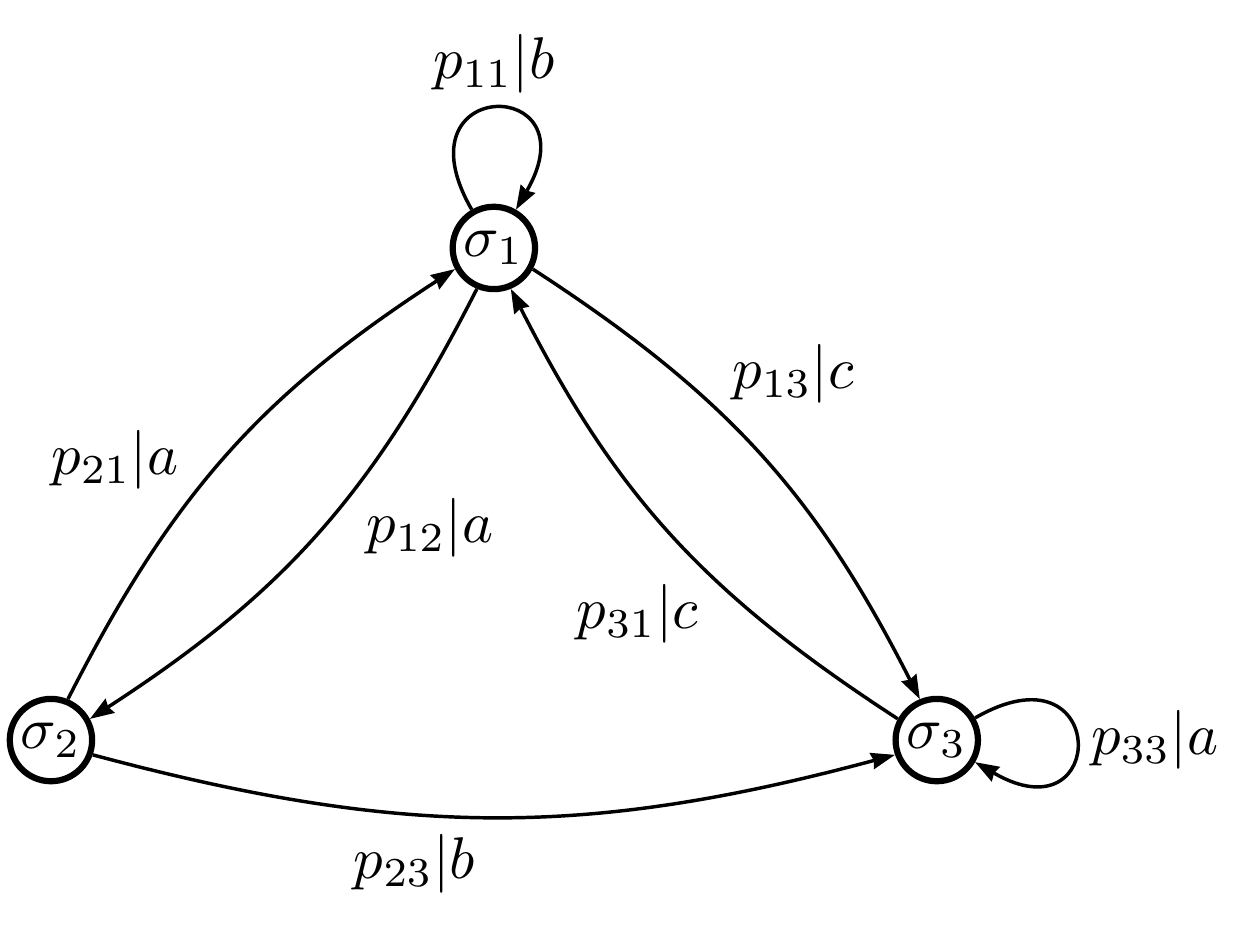} 
\caption{A three-state \eM\ $M$ with alphabet $\MeasAlphabet = \{a,b,c\}$. 
  }
\label{fig:ThreeStateEM}
\end{figure}

\begin{figure}[ht]
\includegraphics[scale=0.45]{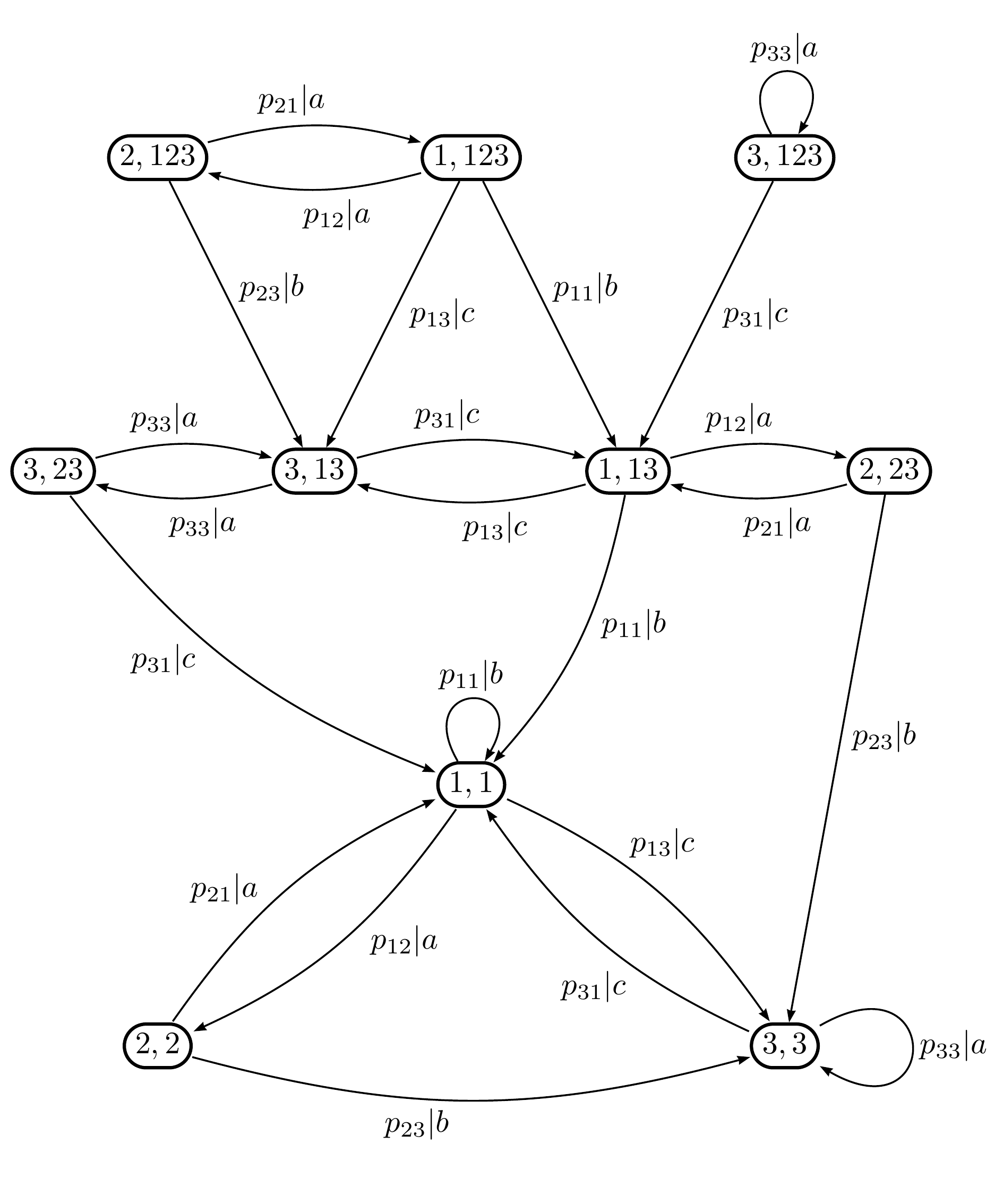} 
\caption{The possibility machine $\Mt$ for the three-state \eM\ $M$ of Fig.
  \ref{fig:ThreeStateEM}. The state names have been abbreviated for display
  purposes: e.g.,
  $(\causalstate_1, \{ \causalstate_1, \causalstate_2, \causalstate_3 \})
  \rightarrow (1,123)$.
  }
\label{fig:TP_Machine}
\end{figure}

\section{}
\label{AppendixB}

We prove Eq. (\ref{Eq:VectorNormLimit}) in Sec.  \ref{SubSec:SyncRate}.
(Restated here as Lemma \ref{Lem:VectorNormLimit}.)
\begin{Lem}
For any exact \eM\ $M$, 
\begin{equation}
\lim_{L \to \infty} \norm{\pi^B B^L}_1^{1/L} = r(B) ~. 
\end{equation}
\label{Lem:VectorNormLimit}
\end{Lem}

In what follows $A$ denotes an arbitrary $m \times m$ matrix and $\vv$ and
$\vw$ denote row $m$-vectors. Unless otherwise specified, the entries of
matrices and vectors are assumed to be complex. 

\begin{Def}
The \emph{(left) matrix $p$-norms} $(1 \leq p \leq \infty)$ are defined as:
\begin{equation}
\norm{A}_p = \max\{ \norm{\vv A}_p: \norm{\vv}_p = 1\} ~.
\end{equation}
\end{Def}

The following facts will be used in our proof.
\begin{Fact}
\label{Fact1}
If $A$ is a matrix with real nonnegative entries and $\vv = (v_1, \ldots, v_m)$
is a vector with real nonnegative entries, then:
\begin{align}
\norm{\vv A}_1 & = \sum_{k = 1}^m \norm{(v_k \ve_k) A}_1 ~,
\end{align}
where $\ve_k = (0,\ldots,1,\ldots,0)$ is the $k_{th}$ standard basis vector. 
\end{Fact}

\begin{Fact}
\label{Fact2}
Let $A$ be a matrix with real nonnegative entries, let $\vv = (v_1,\ldots,v_m)$
be a vector with complex entries, and let
$\vw = (w_1, \ldots , w_m) = (|v_1|, \ldots , |v_m|)$. Then:
\begin{equation}
\norm{\vv A}_1 \leq \norm{\vw A}_1 ~.
\end{equation}
\end{Fact}

\begin{Fact}
\label{Fact3}
For any matrix $A = \{a_{ij}\}$, the matrix $1$-norm is the largest absolute
row sum:
\begin{equation}
\norm{A}_1 = \max_i \sum_{j=1}^m |a_{ij}| ~.
\end{equation}
\end{Fact}

\begin{Fact}
\label{Fact4}
For any matrix $A$, $L \in \N$, and $1 \leq p \leq \infty$:
\begin{equation}
\norm{A^L}_p \leq \norm{A}_p^L ~.
\end{equation}
\end{Fact}

\begin{Fact}
\label{Fact5}
For any matrix $A$ and $1 \leq p \leq \infty$:
\begin{equation}
\lim_{L \rightarrow \infty} \norm{A^L}_p^{1/L} =  r(A) ~,
\end{equation}
where $r(A)$ is the (left) spectral radius of $A$:
\begin{equation}
r(A) = \max \{|\lambda| : \lambda \mbox{ is a (left) eigenvalue of $A$} \}.
\end{equation}
(This is, of course, the same as the right spectral radius, but we emphasize
the left eigenvalues for the proof of Lemma \ref{Lem:VectorNormLimit} below.) 
\end{Fact}

Fact \ref{Fact1} can be proved by direct computation, and Fact \ref{Fact2}
follows from the triangle inequality. Fact \ref{Fact3} is a standard result
from linear algebra. Facts \ref{Fact4} and \ref{Fact5} are finite-dimensional
versions of more general results established in Ref. \cite{Reed80} for bounded
linear operators on Banach spaces.

Using these facts we now prove Lemma \ref{Lem:VectorNormLimit}.

\begin{proof}
By Fact \ref{Fact5} we know:
\begin{equation}
\limsup_{L \to \infty} \norm{\pi^B B^L}_1^{1/L} \leq r(B) ~.
\end{equation}
Thus, it suffices to show that:
\begin{equation}
\liminf_{L \to \infty} \norm{\pi^B B^L}_1^{1/L} \geq r(B) ~.
\end{equation}

Let us define the \emph{B-machine} to be the restriction of the $\Mt$ machine
to its nonreccurent states. The state-to-state transition matrix for this
machine is $B$. We call the states of this machine \emph{B-states} and refer to
paths in the associated graph as \emph{B-paths}. Note that the rows of
$B = \{b_{ij} \}$ are substochastic:
\begin{equation}
\sum_j b_{ij} \leq 1 ~,
\end{equation}
for all $i$, with strict inequality for at least one value of $i$ as long as $M$ has more than 1 state.  

By the construction of the B-machine we know that for each of its states $q_j$
there exists some initial state $q_i = q_{i(j)}$ such that $q_j$ is accessible
from $q_{i(j)}$. Define $l_j$ to be the length of the shortest B-path from
$q_{i(j)}$ to $q_j$, and $l_{max} = \max_j l_j$. Let $c_j > 0$ be the
probability, according to the initial distribution $\pi^B$, of both starting
in state $q_{i(j)}$ at time $0$ and ending in state $q_j$ at time $l_j$: 
\begin{align*}
c_j = ( \pi_{i(j)} \ve_{i(j)} B^{l_j})_j ~.
\end{align*}
Finally, let $C_1 = \min_{j} c_j$.

Then, for any $L > l_{max}$ and any state $q_j$ we have:
\begin{align}
\norm{\pi^B B^L}_1	& \geq \norm{\pi_{i(j)} \ve_{i(j)} B^L}_1 \\
  & = \norm{(\pi_{i(j)} \ve_{i(j)} B^{l_j}) B^{L - l_j} }_1 \\
  & \geq \norm{c_j \ve_j B^{L - l_j}}_1 \\
  & \geq C_1 \norm{\ve_j B^{L - l_j}}_1 \\
  & \geq C_1 \norm{\ve_j B^L}_1 ~.
\end{align}
Equation (B12) follows from Fact \ref{Fact1}. The decomposition in Eq. (B13) is
possible since $L > l_{max} \geq l_j$.
Equation (B14) follows from Fact \ref{Fact1} and the definition of $c_j$.
Equation (B15) follows from the definition of $C_1$. Finally, Eq.
(B16) follows from Fact \ref{Fact3}, Fact \ref{Fact4}, and Eq.
(B11). 

Now, take a normalized (left) eigenvector $\vy = (y_1, ... , y_n)$ of $B$
whose associated eigenvalue is maximal. That is, $\norm{\vy}_1 = 1$,
$\vy B = \lambda \vy$, and $|\lambda| = r(B)$.
Define $\vz = (z_1, ... , z_n) = (|y_1|, ... , |y_n|)$.
Then, for any $L \in \N$:
\begin{align}
\sum_{k=1}^n z_k \norm{\ve_k B^L}_1	& = \norm{\vz B^L}_1 \\
  & \geq \norm{\vy B^L}_1\\
  & = \norm{\lambda^L \vy}_1 \\
  & = |\lambda|^L \cdot \norm{\vy}_1 \\ 
  & = r(B)^L ~,
\end{align}
where Eq. (B17) follows from Fact \ref{Fact1} and Eq. (B18) from Fact
\ref{Fact2}. Therefore, for each $L$ we know there exists some $j = j(L)$ in
$\{1, ... , n\}$ such that:
\begin{equation}
z_{j(L)} \norm{\ve_{j(L)} B^L}_1 \geq \frac{r(B)^L}{n} ~.
\end{equation}
Now, $r(B)$ may be $0$, but we can still choose the $j(L)'s$ such that
$z_{j(L)}$ is never zero. And, in this case, we may divide through by
$z_{j(L)}$ on both sides of Eq. (B22) to obtain, for each $L$:
\begin{align}
\norm{\ve_{j(L)} B^L}_1 & \geq \frac{r(B)^L}{n \cdot z_{j(L)}} \nonumber \\
  & \geq C_2 \cdot r(B)^L,
\end{align}
where $C_2> 0$ is defined by:
\begin{align*}
C_2 = \min_{z_j \not= 0} \frac{1}{n \cdot z_j} ~.
\end{align*}
Therefore, for any $L > l_{max}$ we know:
\begin{align}
\norm{\pi^B B^L}_1 	& \geq C_1 \cdot \norm{\ve_{j(L)} B^L}_1 \\
  & \geq C_1 \cdot  \left( C_2 \cdot r(B)^L \right) \\
  & = C_3 \cdot r(B)^L ~,
\end{align}
where $C_3 \equiv C_1 C_2$. Equation (B24) follows from Eq. (B16) and
Eq. (B25) follows from Eq. (B23). 
Finally, since this holds for all $L > l_{max}$, we have:
\begin{align}
\liminf_{L \to \infty} \norm{\pi^B B^L}_1^{1/L}
  & \geq \liminf_{L \to \infty} \left( C_3 \cdot r(B)^L \right)^{1/L}
  \nonumber \\
  & = r(B) ~.
\end{align}
\end{proof}


\bibliography{ref,chaos}

\end{document}

%% file: cmechabbrev.tex
\newcommand{\eM}     {$\epsilon$\protect\nobreakdash-machine}
\newcommand{\eMs}    {$\epsilon$\protect\nobreakdash-machines}

\newcommand{\EMs}    {$\epsilon$\protect\nobreakdash-Machines}

\newcommand{\Process}{\mathcal{P}}

\newcommand{\MeasAlphabet}	{\mathcal{A}}
\newcommand{\MeasSymbol}   { {X} }
\newcommand{\meassymbol}   { {x} }

\newcommand{\past}	{ {\overleftarrow {\meassymbol}} }

\newcommand{\Future}	{ \overrightarrow{\MeasSymbol} }
\newcommand{\future}	{ \overrightarrow{\meassymbol} }

\newcommand{\CausalState}	{ \mathcal{S} }

\newcommand{\causalstate}	{ \sigma }

\newcommand{\CausalStateSet}	{ \boldsymbol{\CausalState} }

\newcommand{\Prob}      {\Pr} 

\newcommand{\hmu}		{h_\mu}



%% file: shortcuts.tex
\theoremstyle{plain}   \newtheorem{Alg}{Algorithm}
\theoremstyle{plain}   \newtheorem{Lem}{Lemma}
\theoremstyle{plain} 	
\theoremstyle{plain} 	\newtheorem{The}{Theorem}
\theoremstyle{plain} 	\newtheorem{Prop}{Proposition}
\theoremstyle{plain} 	
\theoremstyle{plain}	\newtheorem*{Rem}{Remark}
\theoremstyle{plain}	\newtheorem{Def}{Definition} 
\theoremstyle{plain}	
\theoremstyle{plain} \newtheorem{Fact}{Fact}

\def\norm#1{\|#1\|}

\def\N{\mathbb{N}}

\def\vy{\overrightarrow{y}}

\def\vz{\overrightarrow{z}}

\def\vv{\overrightarrow{v}}
\def\vw{\overrightarrow{w}}
\def\ve{\overrightarrow{e}}

\def\phit{\widetilde{\phi}}

\def\pit{\widetilde{\pi}}

\def\Mt{\widetilde{M}}

\def\Nt{\widetilde{N}}

\def\Tt{\widetilde{T}}

\def\AvgUncertainty{\mathcal{U}}

\def\goesonx{\stackrel{x}{\rightarrow}} 
\def\goesonw{\stackrel{w}{\rightarrow}} 

\newcommand{\SYN}{ \mathrm{SYN} }
\newcommand{\WSYN}{ \mathrm{WSYN} }
\newcommand{\NSYN}{ \mathrm{NSYN} }

\def\L{\mathcal{L}}
\def\LM{\mathcal{L}(M)}

\def\LLM{\mathcal{L}_L(M)}
\def\LiM{\mathcal{L}_{\infty}(M)}

\def\sumw{\sum_{w \in \LLM}}

